\documentclass[12pt]{article}

\usepackage{amsmath,amsthm,amssymb,hyperref,MnSymbol}
\usepackage{a4wide,hyperref}
\hypersetup{colorlinks=true, linkcolor=red,citecolor=blue,urlcolor=blue}

\usepackage{mathrsfs}
\usepackage{authblk}

\newcommand{\nclprob}[3]{\textup{
\begin{center}
\begin{tabular}{|r p{13cm} |}
\hline
\multicolumn{2}{|l|}{\textsc{#1}}\\
\textit{Input:}& #2.\\
\textit{Problem:}& #3\\
\hline
\end{tabular}
\end{center}
}}

\newcommand{\nclpprob}[4]{\textup{
\begin{center}
\begin{tabular}{|r p{13cm} |}
\hline
\multicolumn{2}{|l|}{\textsc{#1}}\\
\textit{Input:}& #2.\\
\textit{YES:}& #3\\
\textit{NO:}&#4\\
\hline
\end{tabular}
\end{center}
}}

\theoremstyle{theorem}
\newtheorem{theorem}{Theorem}
\newtheorem{lemma}[theorem]{Lemma}
\newtheorem{proposition}[theorem]{Proposition}
\newtheorem{corollary}[theorem]{Corollary}

\theoremstyle{definition}
\newtheorem{definition}[theorem]{Definition}
\newtheorem{example}[theorem]{Example}
\newtheorem{remark}[theorem]{Remark}

\newcommand{\F}{\mathbb{F}}

\newcommand{\N}{\mathbb{N}}

\newcommand{\R}{\mathbb{R}}

\newcommand{\ceil}[1]{\left\lceil #1 \right\rceil}

\renewcommand{\b}[1]{\boldsymbol{#1}}

\renewcommand{\le}{\leqslant}
\renewcommand{\leq}{\leqslant}
\renewcommand{\ge}{\geqslant}
\renewcommand{\geq}{\geqslant}

\newcommand{\Hw}{\mathit{Hw}}
\newcommand{\E}{\mathbb{E}}
\newcommand{\Var}{\mathrm{Var}}

\renewcommand{\P}{\mathsf{P}}
\newcommand{\NP}{\mathsf{NP}}
\newcommand{\NE}{\mathsf{NE}}
\newcommand{\Ppoly}{\mathsf{P/poly}}
\newcommand{\FML}{\mathsf{FML}}

\newcommand{\NEXP}{\mathsf{NEXP}}
\newcommand{\NC}{\mathsf{NC}_1}

\newcommand{\NTIME}{\mathsf{NTIME}}

\newcommand{\SIZE}{\mathsf{SIZE}}
\newcommand{\PH}{\mathsf{PH}}

\newcommand{\lex}{\mathit{lex}}

\newcommand{\res}{{\upharpoonleft}}

\newcommand{\AND}{\mathsf{AND}}
\newcommand{\XOR}{\mathsf{XOR}}

\newcommand{\MCSP}{\mathsf{MCSP}}

\newcommand{\PFML}{\mathsf{PFML}}

\begin{document}

\title{\bf Simple  general magnification\\
of circuit lower bounds}

\author{Albert Atserias\thanks{Universitat Polit\`ecnica de Catalunya, and  Centre de Recerca Matem\`atica}\qquad
 Moritz M\"uller\thanks{Universit\"at Passau}}


\maketitle
\begin{abstract} We introduce a technically and conceptually simple 
approach to magnification of circuit and formula lower bounds.
Central to the method are so-called {\em distinguishers}, sparse
matrices that retain some of the key properties of error-correcting
codes.  As applications, we generalize and strengthen known {\em
  general} (not problem specific) magnification results and in
particular achieve magnification thresholds below known lower
bounds. For example, we show that fixed-polynomial formula-size lower
bounds for~$\NP$ are implied by slightly superlinear formula-size
lower bounds for approximating any sufficiently sparse problem
in~$\NP$. We also show that the thresholds achieved are sharp.
Additionally, our approach yields a {\em uniform} magnification result
for the Minimum Circuit Size Problem (MCSP). This seems to sidestep
the localization barrier.
\end{abstract}

\section{Introduction}

We use standard terminology and notation: by a {\em circuit} (resp.~{\em formula}) we mean one with inner gates labeled $\neg,\wedge,\vee$ of fan-in at most 2 (and, resp., of fan-out at most 1); its {\em size} is the number of gates. A {\em promise problem} is given by disjoint sets YES and NO of binary strings. It is contained in $\SIZE[s]$ (resp.~$\FML[s]$), where $s:\N\to\N$, if for all sufficiently large~$n\in\N$ there is a  circuit (resp.~formula) $C_n$ of size $\le s(n)$ 
that accepts all $n$-bit strings in YES and rejects all $n$-bit strings in NO. The class $\P\textit{-uniform-}\SIZE[s]$ additionally requires the map $n\mapsto C_n$ to be computable in time polynomial in $n$.

Note
$\Ppoly=\SIZE[n^{O(1)}]$ and $\NC=\FML[n^{O(1)}]$. 

\subsection{Why are circuit lower bounds difficult?}

The non-uniform version of $\P\neq\NP$ asks  for a problem in $\NP$ that is not decided by circuits of polynomial size, i.e., to show $\NP\not\subseteq\Ppoly$. This seems to be out of reach of current techniques.
Today not even
fixed-polynomial size lower bounds are known, not even for $\NE$ and formulas, i.e., whether $\NE\not\subseteq\FML[n^c]$ for all $c\in\N$, equivalently $\NEXP\not\subseteq\NC$.


Why are circuit lower bounds so difficult? The natural proof barrier \cite{rr} points out some limitation of current techniques -- the notion of a natural proof is, however, informal. Razborov's program \cite{raz1} asks for a formally precise barrier in form of a natural theory that proves known weak lower bounds (cf.~\cite{raz1,mp}) and does not prove strong conjectured ones. The unprovability is, ironically, again dubbed to be out of reach of current techniques -- of proof complexity (cf.~\cite{raz}, \cite[Chapter 27-30]{kra}).

A more direct approach to somehow explain the apparent hardness of circuit lower bounds is to try to establish that the {\em Minimum Circuit Size Problem} is computationally hard in some sense. Namely, for $\sigma:\N\to\N$,  
 \nclprob{$\MCSP[\sigma]$}{$x\in \{0,1\}^{n}$ with $n=2^\ell$ for some $n,\ell\in\N$}{is $x$ computable by a circuit of size $\le \sigma(\ell)$?}
Here, that $x$ is computed by a circuit $C$ means that $C$ has  $\ell$ inputs and outputs the $i$-th bit of $x$ when given the lexicographically $i$-th $\ell$-bit string as input.

Research into the complexity of  $\MCSP[\sigma]$ dates back to the 1950s~\cite{tra} and intensified in recent years, 
starting with \cite{kabanets}. By the idea that its computational hardness reflects the difficulty of proving circuit lower bounds, $\MCSP[\sigma]$ is expected to become harder for larger~$\sigma$ but this is not known.
Concerning its envisioned hardness we refer to \cite{allender} for a survey  and mention here only an obstacle: $\MCSP[2^{o(\ell)}]$ is $2^{n^{o(1)}}$-sparse and 
Buhrmann and Hitchcock~\cite{bh} showed that 
 such problems are not $\NP$-hard unless $\mathsf{PH}$ collapses.
 Recall, for $q:\N\to\N$, a problem $Q$ is {\em $q$-sparse} if $|Q\cap\{0,1\}^n|\le q(n)$ for all $n\in\N$.

\subsection{Magnification}

Thus, not only seem strong circuit lower bounds to be out of reach of current techniques, there is also a lack of understanding of what this very statement means. This is strikingly emphasized  by what
Oliveira and Santhanam \cite{os} called {\em magnification} results: 
a strong lower bound  that appear(ed)  out of reach of current techniques is implied by a lower bound that is ``almost known'' -- meaning that there are known  lower bounds that are only marginally weaker or concern slight variations of the problem or the computational model. 

Such results are interesting, first, because they unveil inconsistencies in our intuitions for what is or not within reach of current techniques.
Second, they provide an approach to strong lower bounds that sidesteps the natural proof  barrier --  see  \cite{beyond} for an insightful  discussion. Both perspectives focus attention on the {\em magnification threshold} -- the ``almost known'' lower bound.

The initial magnification result from \cite{os} concerns a promise problem relaxing  $\MCSP[\sigma]$, namely its approximation version: for  $\epsilon:\N\to\R_{>0}$
\nclpprob{$\epsilon$\text{-}$\MCSP[\sigma]$}{$x\in \{0,1\}^{n}$ with $n=2^\ell$ for some $n,\ell\in\N$}{$x$ is computable  by a circuit of size $\le \sigma(\ell)$.}{$d_H(x,y)\ge \epsilon(n)\cdot n$ for all $y\in\{0,1\}^n$ computable by a circuit of size $\le \sigma(\ell)$.}
Here, $d_H(x,y)$ is the Hamming distance of $x,y$. In other words, NO instances are truth tables of functions that are not $(1-\epsilon(n))$-approximable by circuits of size  $\le \sigma(\ell)$. 

Observe, the smaller~$\epsilon$, the more NO instances, the harder the problem, and $\epsilon\text{-}\MCSP[\sigma]$ is the same as $\MCSP[\sigma]$ if $\epsilon(n)\le 1/n$ for all $n>0$.
Under cryptographic assumptions, it is known that  $1/3$-$\MCSP[2^{\sqrt{\ell}}]\notin\Ppoly$ (cf.~\cite{os}).

\begin{theorem}[\cite{os}]\label{thm:os}  $\NP\not\subseteq\NC$ if
there exists a real $\delta > 0$ such that
$$
n^{-o(1)}\text{-}\MCSP[ 2^{\sqrt{\ell}} ] \not\in \FML[n^{1+\delta}].
$$
\end{theorem}

This magnification threshold  can  be considered to be ``almost known'':
\begin{theorem}[\cite{hs}] \label{thm:hs} $\MCSP[2^{\sqrt{\ell}}]\notin \FML[n^{2-\delta}]$ for all reals $\delta>0$.
\end{theorem}

Following \cite{os} a wide variety of magnification results have been discovered, with magnification thresholds concerning almost formulas \cite{beyond}, probabilistic formulas \cite{sharp},  zero-error heuristics  \cite{os,sparse}, sublinear randomized algorithms \cite{os}, streaming algorithms \cite{mmw}, refuters \cite{os,refuters} or obstruction sets \cite{sharp}.  A few magnification results have also been found in proof complexity \cite{mp,abm}. 
 We do not attempt a survey and  refer to \cite{beyond} for a systematic overview.

A central role is 
played by $\MCSP[\sigma]$ and its variants. Besides the mentioned  approximation relaxation, prominently its gap relaxation
(e.g.~\cite{os,ops,beyond}), its search version \cite{mmw}  and their siblings for time-bounded Kolmogoroff complexity instead circuit complexity. Many proofs of magnification results build on deep results in complexity theory, and 
many are somewhat ad hoc, exploiting special properties of the $\MCSP[\sigma]$ variants.

\subsection{General magnification}

As for a  more principled approach, a significant step forward was taken by
Chen et al.~\cite{sparse}. Only the sparsity matters: ``analogous weak circuit lower bounds [\ldots] for {\bf any} equally-sparse $\NP$ language also imply major separations in circuit complexity.''

We refer to such results as {\em general} magnification. An example predating \cite{sparse} is

\begin{theorem}[\cite{cmmw}] \label{thm:cmmw}
$\NEXP\not\subseteq\Ppoly$ if there exists an $2^{n^{o(1)}}$-sparse problem $Q\in\NP$ such that
$$
Q\not\in \SIZE[n^{1+o(1)}].
$$
\end{theorem}

In \cite{sparse} the conclusion is strengthened to
$\NP\not\subseteq\SIZE[n^c]$ for all $c\in\N$ and a version for
formulas is proved.  A version for probabilistic formulas appears in
\cite{sharp} -- we use the following ad hoc notation: a promise
problem belongs to $\PFML[s]$, where $s:\N\to\N$, if for all large
enough $n\in\N$ there exists a probabilistic formula $\b{F}$ of size
$\le s(n)$ (i.e., a random variable whose values are formulas of size
$\le s(n)$) such that $\Pr[\b{F}(x)=1]=1$ for all YES instances~$x$ of
length $n$ and $\Pr[\b{F}(x)=1]\le 1/4$ for all NO instances $x$ of
length $n$.

\begin{theorem}[\cite{sparse,sharp}] \label{thm:sparse}
$\NP\not\subseteq\FML[n^c]$ for all $c\in\N$
if  there exists an $2^{n^{o(1)}}$-sparse problem $Q\in\NP$ such that (a) or (b):
\begin{enumerate}\itemsep=0pt
\item[(a)] $Q\not\in \FML[n^{3+o(1)}].$
\item[(b)] $Q\not\in \PFML[n^{2+o(1)}].$
\end{enumerate}
\end{theorem}

The proofs of these results are based on deep results in complexity theory: Theorem~\ref{thm:cmmw}
 relies
 on the hardness versus pseudorandomness trade-off \cite{umans} and the Easy Witness Lemma~\cite{ewl}.       The technical core of the proof of Theorem~\ref{thm:sparse} is the construction \cite{sparse} of a highly efficient hash family based on good error correcting codes and expander-walk sampling.
 
The main strength of these results is their generality, as it enhances the prospects of magnification to eventually trigger breakthrough lower bounds.
Their main weakness  is that it is not clear whether their magnification thresholds should be considered ``almost known''. No explicit problem is known outside $\SIZE[5n]$ 
or $\FML[n^3]$, and no explicit $2^{n^{o(1)}}$-sparse problem is known outside $\FML[n^2]$.
For non-sparse problems, however,  
 even subcubic lower bounds are known for probabilistic formulas, in particular for $\MCSP[2^\ell/\ell^4]$ as shown in \cite[Theorem~47]
{beyond} (based on \cite{cklm}). 

For sparse variants, known lower bounds sit sharply below the threshold:

\begin{theorem}[\cite{sharp}] \label{thm:sharplower} $\MCSP[2^{\sqrt{\ell}}]\notin\PFML[n^{2-\delta}]$ for all reals $\delta>0$. 
\end{theorem}

The proof of this result in~\cite{sharp} is based on a sophisticated  construction of a new pseudorandom restriction generator that deserves independent interest.

\begin{remark} 
Both Theorems~\ref{thm:sparse}~(b) and \ref{thm:sharplower} are proved  for probabilistic formulas with {\em two-sided error}. For the magnification result the variant cited here is slightly stronger and implied by Theorem~\ref{thm:intromain} below. We only consider one-sided error in this work.
\end{remark}


\subsection{Distinguishers}

The proof idea for  a magnification result is as follows. Assume $\NP$ has small circuits and construct a tiny circuit for a given problem $Q\in \NP$. The crucial step is an efficient hash of YES instances of $Q$ to very short fingerprints, a set in $\NP$. 
The tiny circuit is then obtained by computing the hash and running  a small circuit on the fingerprint. Typically the hash is randomized, so the tiny circuit is probabilistic. Finally, ``derandomize this construction in an elementary but careful way''~\cite{os}.

Thus,  the core of magnification is compression. 
It is parameterized complexity theory that deve\-loped a deep theory of compression, namely kernelization theory (cf.~\cite{fund}).\footnote{We use the  informal ``compression'' instead of the formal ``kernelization'' because our problems are not parameterized. A suitable parameterization would be the logarithm of the sparsity, so only depend on the input length; this is, however, against the spirit of parameterized complexity theory.}
 In this context the second authors' PhD Thesis \cite{diss} gave a non-uniform hash family for arbitrary sparse problems in order to derive, via some kernelization theory, Burhman and Hitchcock's result \cite{bh}  mentioned earlier. For this, non-uniformity did not harm but magnification does require uniformity. This leads to our main conceptual contribution:

 \begin{definition} Let $n,m\in\N$ and $0<\epsilon<\delta \le 1$ be real.
  View $x,y\in\{0,1\}^n$ as  row vectors in~$\F^n_2$. 
 An {\em $(n,m,\epsilon,\delta)$-distinguisher} is a binary $n\times m$ matrix $D$ such that 
$$
d_H(x,y)\ge \epsilon \cdot n\ \Longrightarrow\ d_H(x D,y D)\ge \delta \cdot m
$$ 
for all $x,y\in \{0,1\}^n$.
 The {\em weight} of $D$ is the maximum Hamming weight of a column of $D$.
\end{definition}

Observe that a $(n,m,\delta,1/n)$-distinguisher $D$ is the same as the generator matrix of a linear code with relative distance $\delta$.  We  trade low weight for larger $\epsilon>1/n$:

\begin{theorem}\label{thm:dist} Let $0<\epsilon\le 1$.
There is an algorithm that given $n\in\N$ computes in time polynomial in $n$  for some $m\le n^7$ an $(n,m,n^{-\epsilon},1/8)$-distinguisher
of weight $\le \ceil{2n^\epsilon}$.
\end{theorem}

Given an $2^{n^{o(1)}}$-sparse $Q$, roughly,  the fingerprint of an $n$-bit string $y$ are $r$ random positions in $yD$. Low weight ensures efficiency of the hash: each bit of the fingerprint is the $\XOR$ of $\le \lceil 2n^\epsilon\rceil$ many input bits, so  computed by a formula of size $O(n^{2\epsilon})$. 
Assume $d_H(y,x)\ge n^{-\epsilon}\cdot n$ for all $x\in Q$. Then the fingerprint equals that of some  $x\in Q$ with probability $\le 2^{n^{o(1)}}\cdot (7/8)^r$ -- this is small  already for  $r\le n^{o(1)}$. This way, a small formula for the set of fingerprints of $x\in Q$ gives a tiny formula distinguishing such~$y$ from all $x\in Q$, i.e., decide $n^{-\epsilon}$-$Q$. 
This notation generalizes $\epsilon\text{-}\MCSP[\sigma]$ to arbitrary $Q$:
 for $\epsilon:\N\to\R_{>0}$,
 \nclpprob{$\epsilon$-$Q$}{$x\in\{0,1\}^n$ for some $n\in\N$}{$x\in Q$.}{$d_H(x,y)\ge\epsilon(n)\cdot n$ for all $y\in Q$.}

\subsection{This work}\label{sec:introresults}

The contribution of the present work is a conceptually modular and technically simple approach to general magnification. It is based on distinguishers and developed from scratch.

To illustrate how magnification is derived from compression,    Section~\ref{sec:simple} gives a highly simplified proof of Theorem~\ref{thm:cmmw}.
Section~\ref{sec:dist} introduces distinguishers and proves Theorem~\ref{thm:dist}. Section~\ref{sec:mag} proves  our main result:  
 
\begin{theorem}\label{thm:intromain}    $\NP\not\subseteq\FML[n^c]$ for all $c\in\N$ if
 there exists a real $\epsilon>0$ and an $2^{n^{o(1)}}$-sparse 
problem $Q\in\NP$ such that (a) or (b):
\begin{enumerate}\itemsep=0pt
\item[(a)] $ n^{-\epsilon}\text{-}Q\not\in \FML[n^{1+2\epsilon+o(1)}]$.
\item[(b)] $n^{-\epsilon}$-$Q\not\in\PFML[n^{2\epsilon+o(1)}]$.
\end{enumerate}
\end{theorem}
This generalizes  Theorem~\ref{thm:sparse}: 
since $1/n$-$Q$ equals $Q$,
Theorem~\ref{thm:sparse} is the special case setting $\epsilon:=1$.
For small $\epsilon$, the magnification thresholds are well below known lower bounds, and in (b) even sublinear.
Note already tiny improvements of magnification thresholds is what magnification is all about.
The proof is based on the hash sketched in the previous subsection which considerably simplifies the hash from~\cite{sparse};
it does, however, not yield the mentioned strengthening of Theorem~\ref{thm:cmmw} in \cite{sparse}.

We find it remarkable that the magnification threshold Theorem~\ref{thm:intromain}~(b) obtained by our generic method turns out to be sharp. Section~\ref{sec:lower} generalizes 
Theorem~\ref{thm:sharplower} with a much simpler proof (but only for one-sided error): 

\begin{theorem}\label{thm:introlower} $n^{-\epsilon}\text{-}\MCSP[2^{\sqrt{\ell}}]\not\in\PFML[n^{2\epsilon-\delta}]$ for all
reals  $  0<\epsilon,\delta \le 1$.
\end{theorem}

An advantage of our simplified approach to magnification is that it works in a uniform setting. In Section~\ref{sec:uni} we prove:

\begin{theorem} \label{thm:parityPintro}  $\P\neq\NP^{\oplus\P}$ if there exist a real $\epsilon>0$ and a function $\sigma(\ell)\le2^{o(\ell)}$ such that $$
n^{-\epsilon}\text{-}\MCSP[\sigma]\not\in \P\textit{-uniform-}\SIZE[n^{1+\epsilon+o(1)}].
$$
\end{theorem}

This magnification threshold can be judged ``almost known'' because Santhanam and Williams~\cite{sw} showed
$\P\not\subseteq\P\textit{-uniform-}\SIZE[n^c]$ for all $c\in\N$.

Theorem~\ref{thm:parityPintro} is interesting in that is seems to sidestep the {\em localization barrier}:
 \cite{beyond} shows that many lower bound techniques {\em localize} in the sense that they still work when the circuits under consideration are enhanced with oracle gates of small fan-in. Such techniques cannot verify the magnification thresholds in Theorem~\ref{thm:intromain} because every $2^{n^{o(1)}}$-sparse problem can be decided by size $\le n^{2+o(1)}$ probabilistic formulas with certain oracle gates of fan-in $n^{o(1)}$ (Corollary~\ref{cor:local}). But Santhanam and Williams' proof does not seem to localize. 
 
This motivates two questions: the first is 
 for a more constructive proof of the lower bound in \cite{sw}, one that actually exhibits an explicit hard problem in $\P$; this is also relevant for applications  in proof complexity -- see~\cite{ok,bm}. Second,
can we find a {\em general} uniform magnification threshold? Say, plug any $2^{n^{o(1)}}$-sparse problem in place of $\MCSP[\sigma]$ above?

\section{An easy example}\label{sec:simple}
We showcase how magnification is derived from compression by giving a very simple proof of a strengthening of Theorem~\ref{thm:cmmw}. 


\begin{theorem}[\cite{cmmw}] For  all $c\in\N$, all reals $\delta>0$ and 
$\gamma<\delta/(c+1)$, and all $2^{n^{\gamma}}$-sparse problems $Q \in \NTIME[2^{n^\gamma}]$, if $Q \not\in \SIZE[n^{1+\delta}]$, then $\NE \not\subseteq\SIZE[n^c]$.
\end{theorem}

\begin{proof} 
Let $\delta,c,\gamma$ accord the statement and assume $\NE\subseteq\SIZE[n^c]$.
Let $Q\in\NTIME[2^{n^\gamma}]$ be $2^{n^\gamma}$-sparse. We show $Q\in\SIZE[n^{1+\delta}]$.

Let $K$ contain the tuples $\langle t,n,1^{\ceil{n^\gamma}},i \rangle$ 
such that $t,n,i$ are natural numbers in binary, $1\le i\le tn$, and there are $x_1,\ldots,x_t\in Q_{n}$ with $x_1<_\lex\cdots <_\lex x_t$ such that
 the $i$-th bit of the concatenation $x_1\cdots x_t\in\{0,1\}^{tn}$ is 1. Here, $<_\lex$ denotes (strict) lexicographic order, and we write $i$ with exactly $\ceil{\log(tn+1)}$ many bits.
Clearly,  $K\in\NE.$
 
We describe small circuits for $Q$. Fix $n\in\N$. 
We consider {\em interesting} inputs to~$K$, namely, $\langle t,n,1^{\ceil{n^\gamma}}, i\rangle$ where $t:=|Q\cap\{0,1\}^n|$
and $1\le i\le tn$. Their length $m=m(n)\le O(n^\gamma)$ depends only on $n$. 
By assumption there is a size $\le m^c$ circuit $C(i)$  deciding whether $\langle t,n,1^{\ceil{n^\gamma}}, i\rangle\in Q$. 
Let $x^* \in \{0,1\}^{tn}$ be the concatenated list of the $t$ strings in $Q \cap \{0,1\}^n$ in $<_\lex$ order. For every interesting input that is in $K$,
the $x_1,\ldots, x_t$ part of its witness must be $x^*$. Thus, $C(i)=1$ 
if and only if the $i$-th bit of $x^*$ is $1$.

To decide whether an input $x\in\{0,1\}^n$ is in $ Q$ we check whether $x$ appears in the list~$x^*$. We employ binary search: this involves $\ceil{\log t}$ comparisons of $n$-bit strings, so $\ceil{\log t}\cdot n$ calls to~$C$. In total, this gives a circuit for $Q$ of size 
$O( \ceil{\log t}\cdot n\cdot m^c)\le O(n^{\gamma+1+\gamma c})$. Since
$\gamma<\delta/(c+1)$ this is $\le n^{1+\delta}$ for large enough $n$.
\end{proof}
A strength of this magnification result is that it is general, a weakness is that its magnification threshold is maybe not ``almost known''. This motivates to study {\em almost formulas},
circuits where only few gates are allowed fan-out $>1$ -- for such formulas even subcubic lower bounds are known (cf.~\cite[Theorem 30]{beyond}). \cite[Theorem~29] {beyond} magnifies slightly superlinear lower bounds for almost formulas and a gap-version of $\MCSP[\sigma]$, improving  
 an earlier result \cite[Theorem~1.4]{ops} for circuits. The proof is quite involved, building on a constructive version of Lipton and Young's anticheckers~\cite{ly}. Compression via anticheckers  also underlies the already mentioned proof complexity magnification result in \cite{mp}.

\section{Distinguishers}\label{sec:dist}

After some preliminaries, Section~\ref{sec:constr} proves Theorem~\ref{thm:dist} which is used to prove  our main result Theorem~\ref{thm:intromain}. Theorem~\ref{thm:parityPintro} requires {\em strongly explicit} distinguishers and we give a stand-alone construction in Section~\ref{sec:strongexplicit}.

\subsection{Preliminaries}\label{sec:distprelim}

%
%
%

For $n\in\N$ we write 
$
[n]:=\{1,\ldots,n\}
$ 
understanding $[0]=\emptyset$. 
Let $k\le n$ be naturals and $F$ a finite set. A random variable $\b{x}$  (with values) in~$F^n$ is {\em $k$-uniform} 
if its $n$ projections are $k$-wise independent and uniform, 
i.e., for every $I\subseteq[n]$ of size at most $k$,
$\b{x}_I$ is uniform  in~$F^{|I|}$; here, for $I=\{i_1<\cdots <i_{|I|}\}$ and $x\in F^n$ 
 $$
 x_I:=x_{i_1}\cdots x_{i_{|I|}}\in F^{|I|}.$$ 

An $n\times m$ matrix $X$ over $F$ is {\em $k$-uniform} if for~$\b{j}$ uniform in~$[m]$, $X^{\b{j}}$ is $k$-uniform; here,~$X^j$ for $j\in[m]$ denotes the $j$-th column of~$X$.

\begin{example}\label{ex:ind} Joffe~\cite{joffe} (see also \cite{cg,aghp})
observed that for a finite field $\F$ the $n\times \F^k$ matrix~$X$ with entry $\sum_{i\in[k]}x_i\cdot \nu^{i-1}$ at row $\nu$ and column $(x_1,\ldots,x_k)$ is $k$-uniform; here, we assume $[n]\subseteq\F$. If $\tilde \F$ is a subfield of $\F$, one gets  a $k$-uniform~$\tilde X$ over~$\tilde \F$ by applying a suitable surjection from $\F$ onto $\tilde\F$ on each entry.

In particular, we represent the field $\F_{2^\ell}$ by $\{0,1\}^\ell$; it is well known that multiplication, addition and multiplicative inverses  can be computed in polynomial time in such a representation~\cite{shoup}. Given $\ell'<\ell$ and
 a $k$-uniform $X$ over $\F_{2^\ell}$ one obtaines a $k$-uniform matrix over $\F_{2^{\ell'}}$ by chopping 
 off  the last $\ell-\ell'$ many bits of each $X_{ij}$.
 \end{example}

 
 

Recall, a {\em (generator matrix of a linear binary) $(n',n,\delta)$-code} is a matrix 
$C\in\F_2^{n\times n'}$ such that $xC,yC$ have Hamming distance $\ge \delta n'$ for all distinct $x,y\in\F_2^{n}$.
We only need a very basic construction of codes because we can allow any rate $n'\le n^{O(1)}$. We describe such a construction, it   just composes a Reed-Solomon code with a Hadamard code -- having concrete parameters eases the presentation later.

\begin{lemma} \label{lem:code} There is an algorithm that, given $n\in\N$, outputs 
an $(n',n,1/4)$-code for some~$n'\le n^4$  in time polynomial in $n$.

\end{lemma}

\begin{proof} 
Choose a natural $\ell>0$ such that  $n\le  2^{\ell}(\ell+1)\le n^2$.
Given an  $n$-bit string $x$ pad it with 0's to length $2^{\ell}(\ell+1)$. Then $x$ determines 
$(x_1,\ldots, x_k)\in(\F_{2^{\ell+1}})^k$ for $k:=2^\ell$.
Let $p_x$ be the polynomial $\sum_{i<k}x_iX^i$ over $\F_{2^{\ell+1}}$. For $y\in \F_{2^{\ell+1}}$ let
$H(y)$ have length $2^{\ell+1}$ and $j$-th bit $ay^\top$ where $a\in\{0,1\}^{\ell+1}$ is the lexicographically $j$-th string (recall we view $(\ell+1)$-bit strings as row vectors in $\F_2^{\ell+1}$). If $y\neq 0^{\ell+1}$, then $H(y)$ has Hamming weight $2^\ell$.
Code~$x$ by the $n':=2^{2\ell+2}$-bit string $x':=H(p_x(y_1))\cdots H(p_x(y_{2^{\ell+1}}))$ where $y_1,\ldots, y_{2^{\ell+1}}$ lists $\F_{2^{\ell+1}}$. If~$x$ is non-zero, then $<k=2^\ell$ many $p_x(y_i)$ are 0 in $\F_{2^{\ell+1}}$. Then $x'$ has Hamming weight $> 2^{\ell}\cdot 2^{\ell+1}/2=2^{2\ell}=n'/4$. Note $n'\le (2^{\ell}(\ell+1))^2\le n^4$ for large~$n$.

Note the code is linear. Its matrix is computable in time $n^{O(1)}$ because  the code of $x$ can be computed in polynomial time.
\end{proof}

\subsection{Existence}
Recall, an $(n,m,\delta,\epsilon)$-distinguisher for $\epsilon\le 1/n$ is an $(m,n,\delta)$-code and we intend to trade larger $\epsilon$ for small weight. We start observing that such a trade-off is possible:

\begin{proposition}\label{prop:dist} 
There exist $c,d\in\N$ such that for every sufficiently large $n\in\N$  and for every real $\epsilon\le 1/(c\log n)$ there exists an $(n,dn,1/5,\epsilon)$-distinguisher of weight at most $1/\epsilon$.
\end{proposition}

\begin{proof}  Let $\Hw(y)$ denote the Hamming weight of $y\in\{0,1\}^n$.
The following is a biased random subset principle:\medskip

\noindent{\em Claim.} Let $0\le p\le1$ and $y\in\{0,1\}^n$. 
Let~$\b{x}$  be the random string in $\{0,1\}^n$ obtained by
independently, for each $i\in[n]$, setting the $i$-th bit to 1 with probability $p$ and to 0 with probability $(1-p)$. Then
\begin{equation}\label{eq:brsp}
\textstyle \Pr[y\b{x}^{\top}=1]=1/2- (1-2p)^{\Hw(y)}/2.
\end{equation}
Write $w:=\Hw(y)$ and $p_w$ for the r.h.s.~of \eqref{eq:brsp}. The claim is trivial for $w=0$.
For $w>0$ the claim follows by  induction on $w$:
\begin{eqnarray*}
&&\Pr[y\b{x}^{\top}=1]=p\cdot (1-p_{w-1})+ (1-p)\cdot p_{w-1}=(1-2p)p_{w-1}+p=p_w.
\end{eqnarray*}

We choose the constants $c,d$ in the course of proof. Fix $\epsilon\le 1/(c\log n)$. Set $p:=1/(2\epsilon n)$ 
and let $\b{x}_1,\ldots, \b{x}_{dn}$ be independent and distributed as~$\b{x}$ above. 
Then $\E[\Hw(\b{x}_j)]=1/(2\epsilon)$ for every $j\in[dn]$. By Chernoff, $\Hw(\b{x}_j)> 1/\epsilon$ with probability $<2^{-1/(d_1\epsilon)}$ for some constant~$d_1$ (independent of $d,c$). Then, with probability  $\ge 1-dn2^{-1/(d_1\epsilon)}$, all~$\b{x}_j$s have Hamming weight $\le 1/\epsilon$. Call this event~$E_0$. If we choose $c> d_1$, 
then  $\epsilon\le 1/(c\log n)$
implies $\Pr[E_0]>1/2$
for large enough $n$ (and all $d$).

For every $j\in[dn]$ and $y\in\{0,1\}^n$ of Hamming weight $w\ge \epsilon n$, we have by \eqref{eq:brsp}
$$
\Pr[y \b{x}_j^{\top}=1]\ge 1/2- e^{- w/(\epsilon n)}/2> 0.3.
$$
Hence, $\E[\sum_{j\in[dn]} y \b{x}_j^{\top}] > 0.3\cdot dn$. By the Chernoff Bound, $\sum_{j\in[dn]} y^{\top}\b{x}^j\ge 0.2\cdot dn$ with probability at least $1-2^{-dn/d_2}$ for some constant $d_2$ (independent of $c,d,d_1$). Choosing~$d> d_2$, this is $\ge 1-1/2\cdot 2^{-n}$. By the union bound, this holds for all  $y$ as above simultaneously with probability  $\ge 1/2$. Call this event $E_1$.
For large enough $n$, the event $E_0\cap E_1$ has positive probability. Let the $dn$ columns of~$D$ consist of corresponding realizations of the   $\b{x}_j$s.
\end{proof}


\subsection{Construction}\label{sec:constr}
For magnification we need explicit distinguishers. Theorem~\ref{thm:dist} follows from the following. 

\begin{theorem} \label{lem:distinguisher} For
  all sufficiently large naturals~$w,n$ and 
  all  reals $\epsilon,\delta >0$ satisfying the assumption
$\delta \leq (1-1/(\epsilon w))/4,$
 there exists an~$(m,n,\epsilon,\delta)$-distinguisher $D$ of weight at
  most~$w$, where~$m \leq n^{7}$. 
  
Moreover, there is an algorithm that given sufficiently large $w<n$ 
outputs in time polynomial in $n+w$ a binary matrix $D$ that is 
such a distinguisher simultaneously for all $\delta,\epsilon>0$ satisfying~$(1-1/(\epsilon w))/4$. 
\end{theorem}

\begin{proof} We can assume $w< n$: otherwise $D:=C$ from Lemma~\ref{lem:code} is a
$(m,n,\delta,\epsilon)$-distinguisher for all $\epsilon>0$ and $\delta\le 1/4$. Note $m\le n^4$ and, trivially, $D$ has weight $\le n\le w$.

We first assume $n$ is a power of 2 and later remove this assumption. 
Let $w< n$ be sufficiently large so that Lemma~\ref{lem:code} applies and gives
 a $(w',w,1/4)$-code $C$ with $w'\le w^4$. Let 
$\epsilon,\delta>0$ be reals satisfying the assumption $(1-1/(\epsilon w))/4$. 
 
The idea is to define a randomized map on $n$-bit strings $x$ as follows: sample $w$ many positions in the string; this determines a $w$-bit string~$y$; apply the code above and output a random bit of the result. This outputs 1 with probability $\ge 1/4$ in the event that $y\neq 0^w$; 
 and this is likely if
$w$ is large enough compared to the Hamming weight of~$x$. We shall show that
the map is implemented by a matrix with a column for each of the random choices. To bound the number of these columns we sample the $w$ many positions not uniformly but only with pairwise independence. Details follow.

Let $X$ be a 2-uniform $w\times n^2$ matrix over $\F_n$ (Example~\ref{ex:ind}). View $\F_n$ with universe~$[n]$.
 Write   the $j$-th column $X^j$ of $X$ as $X^j(1)\cdots X^j(w)$.
 Our randomized map is defined as follows given $x=x(1)\cdots x(n)$:
\begin{enumerate}\itemsep=0pt
\item  sample $j\in[n^2]$  u.a.r.
\item set $y:= y(1)\cdots y(w)$ where $y(i):=x(X^j(i))$ for $i=1,\ldots, w$
\item sample $j'\in[w']$  u.a.r.
\item output the $j'$-th bit of $y C$. 
\end{enumerate}
Note the output for  choices $j$ in line 1 and $j'$ in line 3 equals 
$$
\sum_{i\in[w]} C_{ij'}y(i)=\sum_{i\in[w]} C_{ij'}x(X^j(i))=\sum_{p\in[n]}\sum_{\substack{i\in[w]\\X^j(i)=p}}C_{ij'}x(p)
$$
where the arithmetic is in $\F_2$. Define a binary $n\times n^2w'$-matrix $D$ as follows. We index rows by numbers $p\in[n]$ and  columns by pairs $(j,j')\in [n^2]\times[w']$. Define
\begin{equation}\label{eq:Dentry}
D_{(p,(j,j'))}:=\sum_{\substack{i\in[w]: X^j(i)=p}}C_{ij'}.
\end{equation}
Thus  our randomized
 map outputs a random bit of $x D$. 
For fixed $(j,j')$, the sum \eqref{eq:Dentry} is empty and hence $D_{(p,(j,j'))}= 0$ for all but $\le w$ many $p\in[n]$.
 Thus, $D$ has weight  $\le w$.

Note $m:= n^2\cdot w'\le n^6$. 
To show $D$ is an $(m,n,\delta,\epsilon)$-distinguisher
we show that if~$x$ has Hamming weight $\ge \epsilon n$, then the randomized map outputs~$1$ with probability
$$\ge(1-1/(\epsilon w))/4\ge  \delta.$$

Let the random variable $\b{h}$ be the Hamming weight of the 
 intermediate $w$-bit string $y$. It is the sum of $w$ many pairwise independent  indicator variables 
 each with expectation $\ge \epsilon$.
 By Chebychev's inequality
 $$
 \Pr[\b{h}=0] \leq \mathrm{Var}[\b{h}]/\mathbb{E}[\b{h}]^2 \leq 1/\mathbb{E}[\b{h}]\le 1/(\epsilon w)
$$
since~$\mathrm{Var}[\bf{h}] \leq \mathbb{E}[\bf{h}]$ by pairwise
independence. Thus, $y\neq 0^w$ with probability  $\ge 1-1/(\epsilon w)$. In this case
$y C$ has  Hamming weight  $\ge w'/4$.

Finally, assume~$n$ is not a power of two. Let~$k$ be such that~$2^{k-1} < n < 2^k$. Consider the
distinguisher~$D'$ just shown to exist with
parameters~$n$,~$\epsilon$,~$\delta,w$ reset to~$n' =
2^k$,~$\delta' = \delta$,~$\epsilon' = \epsilon/2,w' = 2w$ --
note~$\delta' \leq (1-1/(\epsilon'w')) /4$
by assumption $(1-1/(\epsilon w))/4$.
$D'$ is a $(n',m',\delta',\epsilon')$-distinguisher with
$m'\le (n')^6\le n^7$ for $n\ge 2^6$.
Let~$D$ be~$D'$ with its
last~$n'-n$ rows removed. If $x\in\{0,1\}^n$  has Hamming weight
$\ge \epsilon n$, then $x0^{n'-n}$ has Hamming weight $\ge \epsilon' n'$. Then 
the Hamming weight of $x D$ equals that of $(x0^{n'-n}) D'$ which is 
$\ge \delta'n'\ge \delta n$.

For the moreover-part, note that the definition of $D$ does not depend on $\delta,\epsilon$. The polynomial time computability  is clear by \eqref{eq:Dentry}, Example~\ref{ex:ind}, and Lemma~\ref{lem:code}.
\end{proof}


We do not know whether one can replace $n^7$ by $O(n)$ in Theorem~\ref{thm:dist}. This is a matter of no concern for our applications.

\subsection{Strongly explicit distinguishers}\label{sec:strongexplicit}

It is possible to modify our distinguisher construction to make it strongly explicit. 
Instead,
we give an alternative construction from scratch by elaborating Naor and Naor's seminal work \cite{nn}.
The wording ``distinguishing'' is from this work.

\begin{theorem}\label{thm:strdist} 
For every large enough $n$ and every rational $0<\epsilon<1$ there exists 
$m\le n^{11}$ and 
an $(n,m,1/16,\epsilon)$-distinguisher  $D$
of weight $\le 8\log(2n)/\epsilon$. 

Moreover, $D$ is {\em strongly explicit:} there is a polynomial time algorithm that given  $n,\epsilon$ as above, and  numbers $i\in[n],j\in[m]$ rejects if $i\notin[n]$ or $j\notin[m]$, 
and otherwise 
outputs~$D_{ij}$.

\end{theorem}

\begin{proof}
Assume first that $n$ is a power of two. We show later how to get rid of this assumption. Using  Example~\ref{ex:ind}, let $U$ be a  $7$-uniform matrix over $\F_2$ with $n$ rows, and let $V$
   be a $2$-uniform matrix over $\F_n$ with $2n$ rows. 
   More concretely,  
   $U$ is obtained 
from a 7-uniform $n\times \F_n^7$ matrix over $\F_n$ by chopping off all but the first bit of each entry;   
    $V$ is obtained from
a 2-uniform $2n\times \F_{2n}^2$ matrix $X$ over $\F_{2n}$ with $X_{i,(a,b)}:=\hat\imath\cdot a+b$ 
where $\hat \imath$ is the $i$-th element of $\F_{2n}$
in lexicographic order; recall, we represent~$\F_{2n}$ by $\{0,1\}^{\log (2n)}$; then $V$ over~$\F_n$ is obtained letting  $V_{i,(a,b)}$ be $X_{i,(a,b)}$ with the last bit chopped off.


Let $\b{u},\b{v}$ be independent and uniformly distributed columns of $U,V$, respectively. 
Let~$K$ be the set of powers of 2 below $n$. 
For  $k\in K$ define  a random string $\b{v}^k$ as the characteristic $n$-bit string of the 
subset of $[n]$ given by the first $2n/k$ components of $\b{v}$.
More formally,  for  $i\in[n]$, let $\hat \imath$ denote the $i$-th element 
of $\F_n$,  and define the $i$-th bit of $\b{v}^k$ by
\begin{equation}\label{eq:vki}
\b{v}^k_i:=\left\{ \begin{array}{ll} 1&\text{if there is }p\in[2n/k]:\b{v}_p=\hat \imath 
\\
0&\text{otherwise.}\end{array}\right.
\end{equation}
Note that $\b{v}^k$ has Hamming weight $\Hw(\b{v}^k)\le 2n/k$ with probability 1. For $x,y\in\{0,1\}^n$
let $x\wedge y\in\{0,1\}^n$ be the obtained by applying  $\wedge$ bitwise, i.e., the $i$-th bit is 1 if and only if both $x$ and $y$ have $i$-th bit 1.
 \medskip

\noindent{\em Claim 1:} $\Pr[1\le\Hw(y\wedge\b{v}^k)\le 7]\ge 1/4$ for all $y\in\{0,1\}^n$ with $k\le \Hw(y)\le 2 k$.\medskip

\noindent{\em Proof of Claim 1:} 
Write $w:=\Hw(y)$. Define $\b{h}_j^k:=y_{\b{v}_j}$ for $j\in[2n/k]$ and set  $\b{h}_k:=\sum_{j\in[2n/k]}\b{h}_j^k$. Then  $\E[\b{h}^k_j]=w/n$, so $\E[\b{h}_k]=2w/k$, and $\Var[\b{h}_k]\le\E[\b{h}_k]$
by pairwise independence. The event $\Hw(y\wedge\b{v}^k)=0$ implies $\b{h}_k=0$. By Chebychev,
the latter has probability $\le \Var[\b{h}_k]/\E[\b{h}_k]^2\le 1/\E[\b{h}_k]\le k/(2w)\le 1/2$.
Since $\b{h}_k\ge \Hw(y\wedge\b{v}^k)$ with probability 1 and  $2w/k\le 4$ the event
$\Hw(y\wedge\b{v}^k)\ge 8$ implies $\b{h}_k\ge 8$ and hence $ |\b{h}_k-2w/k|\ge 4$. By Chebychev, 
 this has probability  $\le \Var[\b{h}_k]/16\le \E[\b{h}_k]/16\le 2w/(16k)\le1/4$.
\hfill$\dashv$\medskip

Define
$
\b{r}^k:=\b{v}^k\wedge\b{u}.
$\medskip

\noindent{\em Claim 2:} $\Pr[y(\b{r}^k)^{\top}=1]\ge 1/8$ for all $y\in\{0,1\}^n$ with $k\le \Hw(y)\le 2 k$.\medskip

\noindent{\em Proof of  Claim 2:} If $z\in\{0,1\}^n$ and $1\le\Hw(z)\le 7$, then $\Pr[z\b{u}^{\top}=1]= 1/2$. Indeed, let 
$I:=\{i\in[n]\mid z_i=1\}$; then $\Pr[z\b{u}^{\top}=1]=\Pr[\Hw(\b{u}_I)\text{ is odd}]$ and this equals $1/2$ because
$\b{u}_I$ is uniformly distributed in $\{0,1\}^{|I|}$.

Since $\Pr[y(\b{r}^k)^{\top}=1]=\Pr[(y\wedge\b{v}^k)\b{u}^{\top}=1]$ and $\b{v}^k$ and $\b{u}$ are independent,  Claim 2 follows from Claim 1.
\hfill$\dashv$\medskip

We can assume that $\epsilon\ge 1/n$.
Let $k_0$ be the maximal element in $K$ that is $\le \epsilon n$, and let $K_0:=\{k\in K\mid  k\ge k_0\}$. Note $k_0\ge \epsilon n/2$. Hence, for every $k\in K_0$ we have
 $2/\epsilon\ge n/k$ and $\Hw(\b{r}^k)\le \Hw(\b{v}^k)\le 2n/k\le 4/\epsilon $ with probability 1.

Let $\b{R}$ be the random $[n]\times K_0$ matrix whose column with index $k\in K_0$ is $\b{r}^{k}$. Let $y\in\{0,1\}^n$ with $\Hw(y)\ge\epsilon n$. Choose $k\in K_0$ such that $k\le \Hw(y)\le 2k$. 
Thus $y(\b{r}^k)^{\top}=1$ and hence $y\b{R}\neq 0$ with probability $\ge 1/8$ by  Claim 2. For~$\b{z}$ uniformly distributed in $\{0,1\}^{K_0}$ (binary vectors indexed by $K_0$) and independent of~$\b{u},\b{v}$ and hence from $\b{r}^k$, we thus have 
$(y\b{R})\b{z}^{\top}=y (\b{R} \b{z}^{\top})=1$
with probability $\ge1/2\cdot 1/8=1/16$.
Further note that, with probability 1,  
\begin{equation}\label{eq:weight}
\Hw(\b{R}\b{z}^{\top})\le |K_0|\cdot 4/\epsilon\le\log n\cdot 4/\epsilon.
\end{equation}

Observe that  $\b{R}\b{z}^{\top}=g(\b{u},\b{v},\b{z})$ for some function $g$. Let the matrix~$D$ have a column  $g(u,v,z)$ for every column  $u$ of $U$, column $v$ of $ V$, and $z\in\{0,1\}^{K_0}$. 
These are 
\begin{equation}\label{eq:m}
m:=|U|\cdot|V|\cdot 2^{|K_0|}= n^7\cdot (2n)^2\cdot n
\end{equation}
many columns.
For $x\in\{0,1\}^n$ with
$\Hw(x)\ge \epsilon n$ we have $\Hw(xD)\ge m/16$. 
Thus,  $D$ is a $(n,m,\epsilon,1/16)$-distinguisher of weight $\le \log n\cdot 4/\epsilon$.

We check that $D$ is strongly explicit. 
Given a row index $i$ and a column index  $j$ we want to compute $D_{ij}$. The column index $j$ determines  the indices of  columns $u,v$ of $U,V$ and a string $z\in\{0,1\}^{K_0}$. The column index of $V$ is $(a,b)$ for some $a,b\in \F_{2n}$. 
We have $D_{ij}$ is the $i$-th entry in $\sum_{k\in K_0:z_k=1} v^k\wedge u$ (sum in $\F_2^n$).
Test whether $u_i=0$ in polynomial time given $i$ and the column  index of $u$ (by evaluating a degree 7 polynomial in $\F_n$). 
If this is the case, then $D_{ij}=0$. Otherwise, $D_{ij}=\sum_{k\in K_0:z_k=1} v^k_i\mod 2$. To determine 
the bits $v^k_i$ recall~\eqref{eq:vki}: check whether there is 
$p\in[2n/k]$ such that 
$v_p=\hat\imath$; recall $\hat\imath\in\{0,1\}^{\log n}$.  For this, compute the $2$ solutions $p\in\F_{2n}$ to the 2 equations  $p\cdot a+b=\hat\imath c$ where $c$ is a bit.

This proves the theorem in the case $n$ is a power of 2. For arbitrary $n$ let $n\le  n'< 2n$ be a power of two, and  $D'$ be a distinguisher for $n'$ as constructed above but for $\epsilon':=\epsilon/2$. Let $D$ consist of the first $n$ rows of $\tilde D$. If $x\in\{0,1\}^n$ has Hamming weight $\ge \epsilon n$, then $\Hw(x0^{n'-n})\ge \epsilon'n'$, so $\Hw(xD)=\Hw((x0^{n'-n})D')\ge n'/16\ge n/16$. Hence, $D$  is a $(n,m,1/16,\epsilon)$-distinguisher. By \eqref{eq:weight}, its weight is $4\log n'/\epsilon'\le 8\log(2n)/\epsilon$.
By \eqref{eq:m} is has $\le n^{11}$ columns for large enough $n$.
\end{proof}

\section{General magnification}\label{sec:mag}

In this section we fix 
\begin{enumerate}
\item[--] $m,w,r:\N\to\N, \ \epsilon,\delta:\N\to[0,1], \ q:\N\to \R$; assume $m(n)\ge n$ for all $n\in\N$.\item[--] $D$ is a function that maps  every sufficiently large natural $n$ to an 
$(n,m(n),\delta(n),\epsilon(n))$-distinguisher $D(n)$ of weight $w(n)$;
\item[--] $Q\subseteq\{0,1\}^*$ is $2^{q(n)}$-sparse.
\end{enumerate}
When $n$ is clear from context we shall often write $m=m(n), w=w(n)$ etc.

\subsection{The kernel}

\begin{definition}\label{df:K} For $x\in\{0,1\}^n$ and $u=(u_1,\ldots,u_{r(n)})\in[m(n)]^{r(n)}$ define
$$
k(x,u):=\langle n,u,b_1,\ldots, b_{r(n)} \rangle
$$
where $b_j=(x D)_{u_j}$ for all $j\in[r(n)]$. 
Define 
\begin{equation*}
K:=K(Q,D,r):=\big\{k(x,u)\mid x\in Q, u\in[m(|x|)]^{r(|x|)}\big\}.
\end{equation*}
\end{definition}

Note that, using some straightforward encoding, $k(x,u)$ is a binary string of length 
independent of $x$ and at most 
$$
10r(n)\log m(n).
$$

\begin{lemma}\label{lem:kernel}
Let $n\in\N, x\in\{0,1\}^n$ and $\b{u}$ be uniform in $[m]^{r}$.  Then
\begin{enumerate}\itemsep=0pt
\item[(a)] if $x\in Q$, then $\Pr[k(x,\b{u})\in K]=1$;
\item[(b)] if $x$ is a NO instance of $\epsilon$-$Q$, then 
$\Pr[k(x,\b{u})\in K]\le 2^q(1-\delta)^r.$
\end{enumerate}
\end{lemma}

\begin{proof} (a) is obvious. For (b), note $k(x,u)=k(y,v)$ implies $u=v$ and $|y|=n$ and $(x D)_{u_j}=(yD)_{u_j}$ for all $j\in[r]$. For $y\in Q$, $d_H(x,y)\ge \epsilon n$, so $d_H(xD,yD)\ge \delta m$. Hence, the event that $k(x,\b{u})=k(y,v)$ for some $v$ has probability at most $(1-\delta)^r$. A union bound over all $\le 2^q$ many $y\in Q\cap\{0,1\}^n$ gives the claim.
\end{proof}

\begin{lemma}\label{lem:KNP} Assume $r,D$ are computable in time polynomial in $n$, and $r(n)\ge n^{\Omega(1)}$. If $Q\in\NP$, then $K\in\NP$.\end{lemma}

\begin{proof} Choose $c\in\N$ such that $r(n)\ge n^{1/c}$ for sufficiently large $n$. Given an input $y$ check it has the form $\langle n,u,b_1,\ldots,b_s\rangle$ for some $s\in\N$ with $s\ge n^{1/c}$. If the check fails, reject. Otherwise $n\le |y|^{c}$, so $r:=r(n), D:=D(n),m:=m(n)$ can be computed in time polynomial in $|y|$. Check $s=r$ and $u\in[m]^{r}$. If the check fails, reject. Otherwise guess $x\in\{0,1\}^n$ and verify $x\in Q$ in nondeterministic time polynomial in $n$, hence in~$|y|$. Verify $b_j=(x D)_{u_j}$ for all $j\in[r]$.\end{proof}

\subsection{Small formulas with local oracles}
\label{sec:localform}

The {\em locality barrier} from \cite{beyond} states that many known circuit lower bounds methods stay true when the circuits are allowed {\em local} oracles, i.e., of small fan-in. ``The fact that existing magnification theorems produce such circuits is a consequence of the algorithmic nature of the underlying proofs'' \cite[Appendix A.2]{beyond}. To stress this aspect we isolate the construction of formulas with local oracles as an own lemma. 

We consider formulas of the form 
$$
\AND_{a}\circ K_{b}\circ \XOR_{c},
$$ 
where $a,b,c\in\R_{\ge 0}$. Such a formula has  a top $\AND$-gate which receives  $\le a$ inputs from  $K$-oracle gates which receive   $\le b$  inputs from $\XOR$-gates which receive  $\le c$ inputs from the input gates and the constant 1; each inner gate has fan-out $\le 1$.

That a promise problem $Q$ is {\em decided by a probabilistic formula of the form $K_{b}\circ \XOR_{c}$ on input length $n$} means that there is a random variable $\b{F}$ whose values are formulas of the form $K_{b}\circ \XOR_{c}$ and such that $\Pr[\b{F}(x)=1]=1$ for all YES instances $x$ of length $n$, and $\Pr[\b{F}(x)=1]\le 1/4$ for all NO instances $x$ of length $n$.

\begin{lemma}\label{lem:localform}  Let $K:=K(Q,D,r)$. 
\begin{enumerate}\itemsep=0pt
\item[(a)] If  $n$ is sufficiently large and $b:\N\to\N$ 
satisfies
\begin{equation}\label{eq:form}
b(n)r(n)\delta(n) - b(n)q(n)\ge n,
\end{equation}
 then $\epsilon$-$Q$ on input length $n$ is decided by a formula of the form $$ \AND_{b(n)}\circ K_{10r(n)\log m(n)}\circ \XOR_{w(n)}.$$
\item[(b)] If $n$ is sufficiently large and
\begin{equation}\label{eq:probform}
r(n)\delta(n) - q(n)\ge 2,
\end{equation}
then $\epsilon$-$Q$  on  input length $n$ is decided by a probabilistic formula of the form 
$$K_{10r(n)\log m(n)}\circ \XOR_{w(n)}.$$
 \end{enumerate}

\end{lemma}

\begin{proof} 
Observe that each bit $b_j$ in $k(x,u)$ is the $\XOR$ of the input bits with  index $i$ such that $D_{iu_j}=1$ and there are at most $w$ many such $i$'s. Hence,
for fixed $u\in[m]^r$ every bit of $k(x,u)$ is computed by an $\XOR$-gate of fan-in at most $w$. Applying a $K$-oracle gate on top of these $\XOR$-gates gives a $K_{10r\log m}\circ \XOR_{w}$ formula $F_u$.  

For $\b{u}$ uniform  in $[m]^r$ we get a random  formula $F_{\b{u}}$ 
such that the event $k(x,\b{u})\in K$ equals the event $F_{\b{u}}(x)=1$. By Lemma~\ref{lem:kernel},
\begin{enumerate}\itemsep=0pt
\item[(Fa)] if $x\in Q$, then $\Pr[F_{\b{u}}(x)=1]=1$;
\item[(Fb)] if $x$ is a NO instance of $\epsilon$-$Q$, then 
$\Pr[F_{\b{u}}(x)=1]\le 2^q(1-\delta)^r.$
\end{enumerate}

Note  $2^q(1-\delta)^r\le 2^{q-r\delta}\le 1/4$ by \eqref{eq:probform}, so $F_{\b{u}}$ witnesses (b). 
To prove (a), let $\bar{\b{u}}=(\b{u}^1,\ldots,\b{u}^b)$ be a tuple of~$b$  independent random variables, each  uniform in $[m]^r$, and set $F_{\bar{\b{u}}}:=\bigwedge_{i\in[b]}F_{\b{u}^i}$. If $x\in Q$, then $\Pr[F_{\bar{\b{u}}}(x)=1]=1$. If $x$ is a NO instance of $\epsilon$-$Q$, then by \eqref{eq:form}
$$
\Pr[F_{\bar{\b{u}}}(x)=1]\le (2^q(1-\delta)^r)^b< 2^{bq- br\delta }\le 2^{-n}.
$$
Fix  values of $\bar{\b{u}}$ so that the resulting formula rejects all NO instances of length~$n$. Clearly, this formula accepts all $x\in Q\cap\{0,1\}^n$. It has the required form.
\end{proof}

\begin{remark} The derandomization argument above shows $\PFML[s(n)]\subseteq\FML[O(ns(n))]$. Hence, (a) implies (b) in Theorems~\ref{thm:sparse}, \ref{thm:intromain}.\end{remark}

\begin{corollary}\label{cor:local}  For all $0<\epsilon\le1$ and all $2^{n^{o(1)}}$-sparse problems $Q$ there is $K\subseteq\{0,1\}^*$ such that on sufficiently large input length $n$, $n^\epsilon$-$Q$ is decided by a probabilistic formula of size $\le n^{2\epsilon+o(1)}$ with a single $K$-oracle gate
 of fan-in $\le n^{o(1)}$.
\end{corollary}

\begin{proof} Assume $Q$ is $q(n)$-sparse for $q(n)\le 2^{n^{o(1)}}$.
For large enough $n$, let $D$ be an $(n,m,1/5,n^{-\epsilon})$-distinguisher of weight~$\le n^{\epsilon}$ with $m\le dn$ for some constant $d$ (Proposition~\ref{prop:dist}). Let $K:=K(Q,D,r)$ for $r(n):=5q(n)+10$. Then (b) above gives a probabilistic formula of the form $K_{10r(n)\log m}\circ\XOR_{n^{-\epsilon}}$. Replace the $\XOR$-gates by formulas of size $O(n^{2\epsilon})$. The $K$-oracle gate has fan-in $10r(n)\log m\le n^{o(1)}$.
\end{proof}

\subsection{General magnification}

The following implies our main result Theorem~\ref{thm:intromain}. It generalizes \cite[Theorem 1.3(3)]{sharp} and \cite[Theorem~1.1(4)]{sharp} which essentially state the special case for $\epsilon=1$.


\begin{theorem}\label{thm:main} For  all $c\in\N$, all reals $\delta,\epsilon >0$ and $\gamma<\delta/c$ and all $2^{n^{\gamma}}$-sparse problems $Q \in
  \NP$, if (a)~$n^{-\epsilon}\text{-}Q\not\in
  \FML[n^{1+2\epsilon+\delta}]$ or
  (b)~$n^{-\epsilon}$-$Q\not\in\PFML[n^{2\epsilon+\delta}]$, then
  $\NP\not\subseteq\FML[n^c]$.
\end{theorem}

 \begin{proof}
Let  $c,\delta,\gamma, \epsilon$ accord the assumption. We can assume $\epsilon\le 1$ because the statement for $\epsilon>1$ is implied by the one for $\epsilon=1$. Set 
$$
q(n):=n^\gamma, \ \delta(n):=1/8, \ w(n):=\ceil{2n^\epsilon}, \ r(n):=17\lceil n^\gamma\rceil.
$$ 
If $n$ is sufficiently large, Theorem~\ref{thm:dist} gives an $(m,n,\delta(n),n^{\epsilon})$ distinguisher of weight $w(n)$ with $m\le n^7$. Let $Q\in\NP$ be $2^{q(n)}$-sparse. Let
$K:=K(Q,D,r)$ be the kernel from Definition~\ref{df:K}. Then $K\in\NP$ by Lemma~\ref{lem:KNP}.
Assume 
\begin{equation}\label{eq:ass}
K\in\FML[n^c].
\end{equation}
 
(a): for $b(n):=\lceil n^{1-\gamma}\rceil$ 
we have 
$br\delta-bq\ge n$ for sufficiently large $n$, so \eqref{eq:form} of Lemma~\ref{lem:localform}~(a) is satisfied. We get a $$
\AND_{\lceil n^{1-\gamma}\rceil}\circ K_{170\lceil n^\gamma\rceil\log m}\circ \XOR_{\ceil{2n^\epsilon}}
$$
formula $F$ deciding  $n^{-\epsilon}$-$Q$ on instances of sufficiently large length $n$.
By \eqref{eq:ass}, the oracle gates 
in $F$ can be replaced by formulas of size $O((n^\gamma\log n)^c)$. The $\XOR$-gates can be replaced by quadratic size formulas, i.e., size $O(n^{2\epsilon})$. The resulting formula is equivalent to $F$ and has size (assuming $n$ is sufficiently large)
\[
O(n^{1-\gamma}\cdot n^{\gamma c}\cdot (\log n)^{c}\cdot n^{2\epsilon})\le n^{1+\delta+2\epsilon}.
\]

(b):  $r\delta-q\ge n^\gamma\ge 2$ 
for sufficiently large $n$, so
 \eqref{eq:probform} of  Lemma~\ref{lem:localform}~(b) is satisfied.
 We get a probabilistic formula of the form
$K_{cn^\gamma\log n}\circ \XOR_{\ceil{2n^\epsilon}}$ for some $c\in\N$. Replacing 
the oracle gates by formulas of size $O((n^{\gamma}\log n)^c)\le o(n^{\delta})$ and the
$\XOR$-gates by formulas of size $O(n^{2\epsilon})$ yields size
$\le n^{2\epsilon+\delta}$.
%
\end{proof}


%
%

\section{Uniform magnification}\label{sec:uni}

Section~\ref{sec:how} generalizes a magnification result for $\MCSP[\sigma]$ from \cite{sharp}. 
This illustrates a use of strongly explicit distinguishers and a certain flexibility of our method. Section~\ref{sec:add} then infers Theorem~\ref{thm:parityPintro}  by a
  modification of the proof.

\subsection{How to use strongly explicit distinguishers}\label{sec:how}

The following generalizes \cite[Theorem 1.3(1)]{sharp} which states the special case for $\epsilon=1$ and 
$\sigma(\ell)=2^{\gamma\ell}$. The generality with relation to $\sigma$ is worth to be stated because, as mentioned in the introduction,  it is unknown whether $\MCSP[\sigma]$ becomes harder for larger~$\sigma$.


\begin{theorem}\label{thm:parityP}
For  all reals $\delta,\epsilon>0$ there is a real $\gamma>0$ such that
for all $\sigma:\N\to\R_{\ge 0}$ with $\sigma(\ell)\le 2^{\gamma\ell}$,
if $n^{-\epsilon}\text{-}\MCSP[\sigma]\not\in\PFML[n^{2\epsilon+\delta}]$,
then $\oplus\P\not\subseteq\NC$.
\end{theorem}

\begin{proof} Let $\delta,\epsilon>0$ and assume $\oplus\P\subseteq\NC$; we can assume $0<\epsilon\le 1$.
For large enough $n$, let
 $D=D(n)$ be a strongly explicit $(n,m(n),1/16,n^{-\epsilon})$-distinguisher of weight $\le 8\log(2n)n^{\epsilon}$ and $m(n)\le n^{11}$  according to Theorem~\ref{thm:strdist}. 
 
 Call a binary string {\em good} if it has  the form
$$
\big\langle n, s,u,b_1,\ldots,b_r\big\rangle
$$
where $n,s,r\in \N$, $n=2^\ell$ for some $\ell\in\N$, $s\le r$, $u\in[m(n)]^r$ and the $b_j$ are bits. Define $K\subseteq\{0,1\}^*$ as the set of such strings such that there exists $x\in\{0,1\}^n$ that is computed by a circuit of size $\le s$ and $b_j=(xD)_{u_j}$ for all $j\in[r]$.\medskip

\noindent{\em Claim:} $K\in\NC$.\medskip

\noindent{\em Proof of the claim:} We first describe a nondeterministic polynomial time algorithm with a $\oplus\P$ oracle. A given input can be checked to be good in polynomial time:
since $D$ is strongly explicit, $u_j\in[m(n)]$ can be checked in polynomial time. 
We guess a circuit $C$ of size $s$ with $\ell$ inputs: this takes  polynomial time since $s\le r$. For each $j\in[r]$ we have to check $b_j=(xD)_{u_j}$ where $x\in\{0,1\}^n$ is the truth table of $C$, i.e., for all $i\in[n]$ we have $x_i=C(\hat\imath)$ where $\hat\imath$ is the lexicographically $i$-th length $\ell$ string. We thus have to check whether there is an odd number of $i\in[n]$ such that  
$C(\hat\imath)\cdot D_{iu_j}=1$. But, since $D$ is strongly explicit, this is a $\oplus\P$-property of $\hat\imath\in\{0,1\}^\ell$, so checked  by the oracle.

Thus, $K\in\NP^{\oplus\P}$. Now argue:
$$
\NP^{\oplus\P}\subseteq \mathsf{RP}^{\oplus\P}\subseteq  \P^{\oplus \P}/\mathsf{poly}\subseteq\oplus\P/\mathsf{poly}
\subseteq\NC;
$$
the 1st inclusion follows from $\NP\subseteq\mathsf{RP}^{\oplus\P}$ by the Valiant-Vazirani Lemma,
the 2nd from Adleman's trick, the 3rd from $\oplus \P^{\oplus\P}=\oplus\P$ \cite{papa}, and the 4th from our assumption.\hfill$\dashv\medskip$

Choose $c\in\N$ such that $K\in\FML[n^c]$ and set $\gamma:=\delta/(4c)$. Let $\sigma(\ell)\le 2^{\gamma\ell}$ be given. We have to show that $n^{-\epsilon}$-$\MCSP[\sigma]\in\PFML[n^{\delta+2\epsilon}]$.

Set $q(n):=n^{2\gamma}$
and note $\MCSP[\sigma]$ is $2^{q(n)}$-sparse. 
Let $\tilde r(n):=\ceil{16q(n)+32}$ and set $\tilde K:=K(\MCSP[\sigma],D,\tilde r)$ according to Definition~\ref{df:K}. This satisfies  \eqref{eq:probform} of Lemma~\ref{lem:localform}, so $n^{-\epsilon}$-$\MCSP[\sigma]$ has probabilistic formulas of the form
\begin{equation*}\label{eq:uni}
\tilde K_{O(n^{2\gamma}\log n)}\circ\XOR_{O(n^{\epsilon}\log n)}
\end{equation*}
on inputs of sufficiently large length $n=2^{\ell}$. We want to replace the  $\tilde K$-oracle gates by $K$-oracle gates. To this end, note that, by the proof of Lemma~\ref{lem:localform}, the $\XOR$-gates produce an input to $\tilde K$ of the required form and this is a good string (with $r=\tilde r(n)$) except that some~$s$ is missing.
 Hence, we just have to add some  gates producing $s:=\sigma(\ell)$. Note $s\le \tilde r(n)$, so the result is good.
 
Replace the $K$-oracle and $\XOR$ gates by formulas of size $O((q(n)\log n)^c)\le O(n^{\delta/2})$ and  $O(n^{2\epsilon}\log^2n)$. This yields a formula of size $\le n^{2\epsilon+\delta}$ for sufficienly large~$n$.
\end{proof}

\begin{corollary} 
$\oplus\P\not\subseteq\NC$ if
 there exist $\epsilon>0$ and a  function
 $\sigma(\ell)\le 2^{o(\ell)}$ such that
$$
n^{-\epsilon}\text{-}\MCSP[\sigma]\not\in\PFML[n^{2\epsilon+o(1)}].
$$
\end{corollary}

\subsection{Adding uniformity}\label{sec:add}

The following implies Theorem~\ref{thm:parityPintro}.


\begin{theorem} For  all reals $\delta,\epsilon>0$ there is a real $\gamma>0 $ such that for all $\sigma:\N\to\R_{\ge 0}$ with $\sigma(\ell)\le 2^{\gamma\ell}$, if $n^{-\epsilon}\text{-}\MCSP[\sigma]\not\in \P\textit{-uniform-}\SIZE[n^{1+\epsilon+\delta}]$, then $\P\neq\NP^{\oplus\P}$.
\end{theorem}

\begin{proof} Assume $\P=\NP^{\oplus\P}$ and proceed as in the proof of Theorem~\ref{thm:parityP}.
The proof of the claim now gives $K\in\P$. Choose $c\in\N$ such that $K\in\P\textit{-uniform-}\SIZE[n^c]$
and set  $\gamma=\delta/(4c)$, as before.
As seen, Lemma~\ref{lem:localform} gives a probabilistic formula of the form
$$
K_{O(n^{2\gamma}\log n)}\circ\XOR_{O(n^{\epsilon}\log n)}
$$
deciding $n^{-\epsilon}$-$\MCSP[\sigma]$. Recalling the proof of Lemma~\ref{lem:localform}, the formula equals $F_{\b{u}}$ for $\b{u}$ uniform in $[m]^{\tilde r}$ where 
$m,\tilde r$ are as in the proof of Theorem~\ref{thm:parityP}.
 It is clear that the map $u\mapsto F_u$ is computable in polynomial time.

Let $\bar{\b{u}}=(\b{u_1},\ldots,\b{u}_n)$ be uniform in $([m]^{\tilde r})^n$. Then $F_{\bar{\b{u}}}:=\bigwedge_{i\in[n]}F_{\b{u}_i}$ accepts all YES instances and rejects NO instances with probability $\le 4^{-n}<2^{-n}$. Hence there is a realization~$\bar u$ of $\bar{\b{u}}$ such that $F_{\bar u}$ rejects all NO instances. It is not hard to see that such $\bar u$ can be computed from $n$ in time polynomial in $n$ with the help of a $\PH$ oracle. But our assumption implies $\PH=\P$,
so $F_{\bar u}$ can be computed in time polynomial in $n$. 

Compute a circuit to replace the  $K$-oracle gates and a linear size circuit  to replace the $\XOR$-gates. The resulting circuit
is computable in time polynomial in $n$ and has size
\[
O(n\cdot (n^{2\gamma}\log n)^c\cdot n^\epsilon\log n)\le n^{1+\delta+\epsilon}.\qedhere
\]
\end{proof}

\paragraph{Discussion} For $\P$-uniformity the distinguishers are not required to be strongly explicit. 
But the $\P$-uniform version of Theorem~\ref{thm:main} is void: its conclusion $\NP\not\subseteq\P\textit{-uniform-}\FML[n^c]$ is known \cite{sw}, as   discussed in the introduction. 
Thus, contrary to our focus in the introduction on the magnification threshold,
strengthening the conclusion of Theorem~\ref{thm:main} could yield 
 {\em general} uniform magnification as asked for in the introduction.

\section{A sharp lower bound}\label{sec:lower}

This section proves Theorem~\ref{thm:introlower}. The method of proof follows Hirahara and Santhanam's \cite{hs} which in turn builds on~\cite{imz}. We start recalling what is needed.

\subsection{Preliminaries}

Let $n\in\N$. 
A {\em random restriction (on $[n]$)} $\b{\rho}$ is a random variable in $\{0,1,*\}^n$, the set of {\em restrictions (on $[n]$)};
we write restrictions $\rho$ as functions $\rho:[n]\to\{0,1,*\}$.
For $p\in [0,1]$ and $k\in[n]$, $\b{\rho}$ is {\em $(p,k)$-regular} if the projections $\b{\rho}(1),\ldots, \b{\rho}(n)$ are $k$-wise independent and each takes values $*,0,1$ with probabilities $p,(1-p)/2,(1-p)/2$.

For a Boolean function $f$ with $n$ variables and a restriction $\rho$ on $[n]$, the function
$f\res\rho$ is $f$ with $i$-th variable fixed to $\rho(i)$ for all $i\in[n]$ with $\rho(i)\neq*$. For another 
restriction $\rho'$ observe  $(f\res\rho)\res\rho'=f\res \rho\rho'$ where $\rho\rho'$ maps $i$ to $\rho(i)$ if $\rho(i)\neq*$, and to $\rho'(i)$ otherwise.


\begin{example}\label{ex:regularrho} Let $k,2^r\le n$ and $X$ be an $n\times \F_{2^r}^k$-matrix over $\F_{2^{r}}$ according to
Example~\ref{ex:ind}. Assume $r=r'+1$ and set $p:=2^{-r'}$. Choose
$f:\F_{2^{r}}\to\{0,1,*\}$ such that $*$ has $p2^{r}=2$ preimages
and $0,1$ both have $(1-p)2^r/2=2^{r'}-1$ preimages.
Define  $\b{\rho}$ to be a uniformly chosen column of $X$ with entries replaced using $f$. Then $\b{\rho}$ is a $(p,k)$-regular random restriction on $[n]$.
\end{example}

By $L(f)$ we denote the minimal number of leafs of (the formula tree of) a formula computing $f$.
The following is a version of \cite[Lemma~27]{hs} (there attributed to \cite{imz}).

\begin{lemma}\label{lem:hs1} There exists a real~$c \geq 1$ such that for every
$n\in\N_{>0}$, every $f:\{0,1\}^n\to\{0,1\}$,
every~$p \in [0,1]$, every natural~$k \geq 1/p^2$ and
  every $(p,k)$-regular random restriction~$\b{\rho}$ on~$[n]$:
$$\mathbb{E}[L(f\res {\b{\rho}})] \leq
  \max\{c p^2 L(f), c\}.$$
  
\end{lemma}
  
\begin{remark} \cite[Lemma~27]{hs} has $O(p^2 L(f))$ instead $\max\{ cp^2 L(f), c\}$ and an additional assumption $L(f)\ge p^{-2}$. 
In the proof, this assumption is used to bound $\mathbb{E}[L(g\res {\b{\rho}})]$ by $O(p^2 L(g))$ for 
certain functions $g$ satisfying  $L(g)\ge p^{-2}/6$ by the additional assumption; indeed, then $O(p^2 L(g))\ge O(p^2 L(g)+p\sqrt{ L(g)})\ge \mathbb{E}[L(g\res {\b{\rho}})]$ as proved in  \cite{tal} (improving~\cite{hastad}).
Our version uses the bound $\max\{ cp^2 L(g), c\}$ instead.
\end{remark}


A random restriction $\b{\rho}$ is a {\em composition of~$t$ many~$(p,k)$-regular  restrictions on~$[n]$} if $\b{\rho}=\b{\rho}_1\cdots \b{\rho}_t$ for independent $(p,k)$-regular  restrictions $\b{\rho}_1,\ldots,\b{\rho}_t$ on $[n]$. Note that such $\b{\rho}$ is $(p^t,k)$-regular.
The following is a variant of \cite[Theorem~28]{hs} (there attributed to \cite{imz}).

\begin{lemma}\label{lem:hs2}
  There exists a real~$c \geq 1$ such that for all~$t,n \in\N_{>0}$, every $f:\{0,1\}^n\to\{0,1\}$,
  every
  real~$p$ with~$0 < p < 1/(2c)^{1/2}$, every
  natural~$k \geq 1/p^{2}$, 
  and every 
  composition $\b{\rho}$ of~$t$ many~$(p,k)$-regular random
   restrictions on $[n]$:
 $$
 \mathbb{E}[ L(f\res\b{\rho}) ] \leq c^t  p^{2t} L(f) + 2c.
  $$  
\end{lemma}

\begin{proof}
Let~$c$ be the constant from Lemma~\ref{lem:hs1}. We proceed by induction on $t$ and prove
a slightly stronger claim with $2c$ replaced by $c/(1-cp^2)$ (which is $\le 2c$ by $p< 1/(2c)^{1/2}$).

The case~$t=1$ follows from Lemma~\ref{lem:hs1}  
because $\max\{ cp^2 L(f), c\} \leq cp^2 L(f) + c/(1-cp^2)$.
For the inductive step, let $t>0$ and recall 
$\b{\rho}=\b{\rho}' \b{\rho}_t$ for $\b{\rho'}:=\b{\rho}_{1}\cdots\b{\rho}_{t-1}$ where the~$\b{\rho}_i$ are independent and $(p,k)$-regular. 
Letting $\rho'$ range over realizations of $\b{\rho}'$, 
\[
\begin{array}{rcll}
\mathbb E[L(f\res\b{\rho})]&=&\textstyle \sum_{\rho'}\Pr[\b{\rho}'=\rho' ]\cdot \mathbb E[ L((f\res\rho')\res\b{\rho}_t) ] & \; \text{as $\b{\rho}',\b{\rho}_t$ are independent}\\
&\le&\textstyle \sum_{\rho'}\Pr[\b{\rho}'=\rho' ]\cdot \big(  cp^2L(f\res\rho')+c\big)& \; \text{by Lemma~\ref{lem:hs1}}\\
&=&  cp^2  \mathbb E[L(f\res\b{\rho}')]  + c&\\
&\le&   cp^2 \big(    c^{t-1} p^{2(t-1)} L(f) + c/(1-cp^2)    \big)+ c& \; \text{by induction}\\
&=&  c^t  p^{2t} L(f) + c/(1-cp^2).&
\end{array}
\qedhere
\]
\end{proof}

\subsection{The lower bound}

The following implies Theorem~\ref{thm:introlower}. 

\begin{theorem}\label{thm:lower} There exists  $d\in\N_{>0}$ such that for all reals  $  0<\epsilon,\delta\le 1$ and all
$\sigma:\N\to\R_{\ge 0}$ with  $\ell^d\le\sigma(\ell)\le 2^{o(\ell)}$,
$$n^{-\epsilon}\text{-}\MCSP[\sigma]\not\in\PFML[n^{2\epsilon-\delta}].
$$ 
\end{theorem}

\begin{proof} 
We choose $d>0$ in the end of the proof. Let $\delta,\epsilon,\sigma$ be as stated.
Write $2^\ell=n$ and $\sigma=\sigma(\ell)$.
It suffices to rule out probabilistic formulas of size at most
$$s:=n^{2\epsilon-6\delta}.$$

We want to choose $p,t,k$ in Lemma~\ref{lem:hs2} to ensure that $L(F\res\b{\rho})$ is probably constant. Let~$c$ be the lemma's constant, and set 
$$
k:=\ceil{\sigma^{1/d}},\ t:=\ceil{\frac{2d\log s}{\log \sigma}}.
$$ 
and 
for $p$ choose a negative power of 2 such that
\begin{equation*}\label{eq:p}
(c^{1/2}s^{1/(2t)})^{-1}/2\le p\le (c^{1/2}s^{1/(2t)})^{-1}.
\end{equation*}
This ensures $c^tp^{2t}\le 1/s$. As $p\le 1/(2c)^{1/2}$ and $1/p^2\le 4cs^{1/t}\le 4c\sigma^{1/(2d)} \le k$ for large enough~$n$, the lemma applies and gives $\mathbb E[L(F\res\b{\rho})]\le 1+2c$. 
Thus, for a suitable constant  $s_0\in\N$, 
 \begin{eqnarray}\label{eq:L}
  &&\Pr[L(F\res\b{\rho})\ge s_0]<1/2.
\end{eqnarray}
 Given $\b{\rho}$ and an independent $\b{z}$ uniform  in $\{0,1\}^n$ define the random variable $\b{x}$ in $\{0,1\}^n$ to have $i$-th bit equal to $\b{\rho}(i)$ if $\b{\rho}(i)\neq *$, and otherwise equal to the $i$-th bit of $\b{z}$.

Let $E$ be the event that $\b{x}$
 is  {\em not} a NO instance of $n^{-\epsilon}$-$\MCSP[\sigma]$, and $\bar E$ its complement.

\medskip

\noindent{\em Claim:} $\Pr[E]\le o(1)$.

\medskip

\noindent{\em Proof of the claim:} Set $\b{S}:=\b{\rho}^{-1}(*)$. Since $\b{\rho}$ is $(p^t,k)$-regular, $|\b{S}|$ has expectation $p^tn$. Using $(4c)^{t/2}\le s^{o(1)}$,
\begin{equation}\label{eq:pn}
p^tn\ge n/((4c)^{t/2}s^{1/2})\ge n\cdot s^{-1/2-o(1)}\ge n^{1-\epsilon+3\delta-o(1)}\ge 4n^{1-\epsilon+2\delta}.
\end{equation}
The variance of $|\b{S}|$ is $\le \mathbb E[|\b{S}|]$ by $k\ge 2$-wise independence. 
Let $E_0$ be the event that $|\b{S}|\ge  2n^{1-\epsilon+2\delta}$, and $\bar E_0$ be its complement. By Chebychev,  
$\Pr[\bar E_0]\le 4/ \mathbb E[ | \b{S} | ] \le o(1).$
Thus, it suffices to show $\Pr[E\mid E_0]\le o(1)$. 

Let $S$ range over realizations of $\b{S}$
with $|S|\ge 2n^{1-\epsilon+2\delta}$, and let $y$ range over YES instances. Then
 $E\subseteq \bigcup_y E_y$ where  $E_y$ is the event that
$d_H(\b{x}_{\b{S}},y_{\b{S}})< n^{1-\epsilon}$. Since there are at most $2^{n^{o(1)}}$ many YES instances,
   $\Pr[E\mid E_0]\le  2^{n^{o(1)}}\cdot \max_y\Pr[E_y\mid E_0]$. 
 Hence, it suffices to show $\Pr[E_y\mid E_0]\le 2^{-n^{1-\epsilon+\delta}}$ for all $y$.
 Since $\Pr[E_y\mid E_0]\le \max_S\Pr[E_y\mid \b{S}=S]$ it suffices to show
$$ 
\Pr[E_y\mid \b{S}=S]\le 2^{-n^{1-\epsilon+\delta}},
$$
  for all $y,S$. To see this, note $\Pr[E_y\mid\b{S}=S]=\Pr[d_H(\b{z}_{S},y_{S})< n^{1-\epsilon}]$ (by independence of $\b{z}$ and~$\b{S}$) and
$d_H(\b{z}_{S},y_{S})$ is the sum of $|S|$ many independent indicators each with expectation~$1/2$. 
By Chernoff, $\Pr[d_H(\b{z}_{S},y_{S})< n^{1-\epsilon}]\le 2^{-\Omega(|S|)}\le2^{-n^{1-\epsilon+\delta}}$. \hfill$\dashv$\medskip

For contradiction, assume $\b{F}$ is a size $\le s$ probabilistic formula for $n^{-\epsilon}$-$\MCSP[\sigma]$. Then, letting $x$ range over NO instances,
$$\textstyle
\Pr[\b{F}(\b{x})=1]\le \Pr[ E]+\sum_x\Pr[\b{F}(x)=1, \b{x}=x]\le o(1) + 1/4\Pr[\bar E]<1/2.
$$

We thus find realizations $F,\rho,z,x$ of $\b{F},\b{\rho},\b{z},\b{x}$ such that $F(x)=0$ and 
and $L(F\res\rho)< s_0$ (by~\eqref{eq:L}). 
We get a contradiction by finding a YES instance $x'$ rejected by $F$.

Let $x''$ be $\rho$ with $*$ replaced by $0$. Then define $x'$ from $x''$ by flipping $0$s to 1s in positions corresponding to variables in $F\res\rho$ that are 1 in $z$. Then $F(x')=(F\res\rho)(x)=F(x)=0$.
Note that $< s_0$ many positions are flipped.
We are left to show that $x'$ is computed by a circuit of size $\le \sigma$. 

First consider $x''$. We ask for a small circuit  with~2 output bits that computes $\rho$ in the encoding $01,11,00$ of $0,1,*$. Then, taking the conjunction of the output
bits gives a small circuit for $x''$.
To get such a circuit we choose $\b{\rho}$ concretely as the composition of $t$ many $(p,k)$-regular $\tilde{\b{\rho}}$ according  to Example~\ref{ex:regularrho} -- this can be done because  $p=2^{-r}$ for some $r\in\N$ with $2^{r+1}\le n$; indeed, 
$2^{r+1}=2/p\le  4 c^{1/2}\sigma^{1/(4d)} \le k\le  n$ for large enough $n$.

We claim there is a circuit of size polynomial in $k\ell$.  Since $t< \ell$, it suffices to show each of the corresponding realizations $\tilde{\rho}$ is computed  by a small circuit.
 Each $\tilde{\rho}$ is determined by some $j\in \F_{2^{r+1}}^k$, 
 and  $i\mapsto \tilde{\rho}(i)$ is computed 
 (given $k,2^{r+1}\le k,n=2^\ell,j$)    
 in time  polynomial in $k\ell$.   Choose $e\in\N$ such that $x''$ is computed by a circuit $C''$ of size $\le (k\ell)^{e}$. 
 
To get a circuit $C'$ for $x'$ we change $< s_0$ values of $C''$. By a look-up table, $C'$ has size $\le  (k\ell)^{e}+10s_0\ell$. As $k,\ell\le \ceil{\sigma^{1/d}}$, we can choose $d\in\N$ such that~$C'$ has size $\le \sigma$.
\end{proof}

\begin{remark} Given a real $c>0$ we can choose $d$ sufficiently large, so that the YES instance~$x'$ constructed has circuit complexity $\sigma^{1/c}$. Hence, the above holds for the gap / approximation version of $\MCSP$ whose YES instances are strings computed by circuits of size  $\le \sigma^{1/c}$ and whose NO instances are strings  that cannot be $(1-n^{-\epsilon})$-approximated by circuits of size~$\le \sigma$.
\end{remark}

\paragraph{Discussion} \cite[Theorem~1.5]{sharp} proves a slight strengthening of the special case 
$\epsilon=1$.
%
As mentioned in the introduction, the proof method in \cite{sharp} is much more complicated and deserves independent interest. It can handle two-sided error and  for $\sigma(\ell)=2^{\gamma\ell}$  achieves even slightly superquadratic lower bounds~$n^{2+\gamma/2}$ \cite[Theorem~1.6]{sharp} (based on \cite{cklm}). This outperforms the simpler method here. We do not know whether it extends to  $\epsilon<1$.

Concerning lower bounds for the approximation version,  \cite[Theorem~4.7]{os} extended Theorem~\ref{thm:hs} to 
$\epsilon(n)$-$\MCSP[2^{\sqrt{\ell}}]$ but only for  tiny $\epsilon(n)\le n^{-1+\Omega(1/\sqrt{\log n})}$. 
\cite{ops}~states that ``existing lower bound methods
are more suitable for proving lower bounds'' \cite[p.6]{ops} for the gap version as opposed to the approximation version of $\MCSP[\sigma]$. We doubt this, given Theorem~\ref{thm:lower} and its proof. Maybe this is an interesting insight.

\end{document}